\documentclass[conference]{IEEEtran}
\IEEEoverridecommandlockouts
\usepackage{amssymb,amsfonts,amsmath,amsthm}
\usepackage[dvips]{graphicx}
\usepackage{overpic}
\usepackage{subfigure}
\usepackage{color}
\usepackage{cases}
\usepackage{enumerate}
\usepackage{framed}
\definecolor{shadecolor}{rgb}{0.92,0.92,0.92}
\usepackage{listings}
\usepackage{graphicx}
\usepackage{float}
\usepackage{mathrsfs}
\usepackage{pgf}
\usepackage{tikz}
\usetikzlibrary{arrows, automata, positioning, calc, shapes}
\usepackage{comment}
\usepackage[numbers,sort&compress]{natbib}
\usepackage{caption}
\usepackage{hyperref}
\usepackage{bm} 
\usepackage[ruled,linesnumbered]{algorithm2e}
\usepackage{multirow}

\newtheorem{theorem}{Theorem}

\newtheorem{example}{Example}
\newtheorem{definition}{Definition}
\newtheorem{lemma}{Lemma}
\newtheorem{remark}{Remark}
\newtheorem{construction}{Construction}
\newtheorem{problem}{Problem}

\newtheorem{corollary}{Corollary}
\newtheorem{claim}{Claim}

\newcommand{\bd}{\bm{d}}

\newcommand{\br}{\bm{r}}

\newcommand{\bv}{\bm{v}}

\newcommand{\bx}{\bm{x}}

\newcommand{\cC}{\mathcal{C}}
\newcommand{\cD}{\mathcal{D}}

\newcommand{\cF}{\mathcal{F}}

\newcommand{\cS}{\mathcal{S}}

\newcommand{\cX}{\mathcal{X}}

\newcommand{\bbE}{\mathbb{E}}

\usepackage{comment}
\allowdisplaybreaks
\title{Sequential Random Sampling PIR with Multiple Colluding Servers in DNA-Based Data Storage
\thanks{The research was Funded by the European Union (ERC, DNAStorage, 101045114 and EIC, DiDAX 101115134). Views and opinions expressed are however those of the authors only and do not necessarily reflect those of the European Union or the European Research Council Executive Agency. Neither the European Union nor the granting authority can be held responsible for them. The research of C. Wang was supported in part at the Technion by a fellowship from the Lady Davis Foundation.}} 
\author{\textbf{Chen Wang}, \textbf{Natalia Silberstein}, and \textbf{Eitan Yaakobi}\\
\IEEEauthorblockA{Department of Computer Science, Technion---Israel Institute of Technology, Haifa 3200003, Israel}
{Emails: cwang@campus.technion.ac.il, natalys@cs.technion.ac.il, yaakobi@cs.technion.ac.il}
}
\date{\today}
\begin{document}
\maketitle
\begin{abstract}
    As DNA-based data storage evolves, protecting user privacy during data retrieval has become increasingly important. We study sequential random sampling DNA private information retrieval (SRS DNA PIR) with multiple colluding random sampling servers, where the database is partitioned into servers of equal size. We investigate the tradeoff between the \textit{download cost}, defined as the expected number of queries, and the \textit{privacy leakage}, measured by mutual information. We derive lower bounds on this tradeoff, including a bound given by an optimization problem. This bound is tight when each server stores two files, and we construct schemes that attain it. For servers of any size, we construct schemes that apply a single-server scheme to a randomly selected subset of servers. 
\end{abstract}

\section{Introduction}
With the rapid growth of digital data, the demand for efficient and reliable storage media continues to increase~\cite{rydning2022worldwide}. DNA-based data storage is a promising approach due to its high information density and durability~\cite{church2012next,goldman2013towards}. In such systems, data is encoded into synthetic DNA strands and recovered through sequencing~\cite{anavy2019data,organick2018random,yazdi2017portable,erlich2017dna}. The cost and latency of sequencing motivate efforts to reduce the number of reads required for recovery~\cite{chandak2019improved,preuss2024sequencing,daniella2023cover}. Since strands are stored in unordered pools, random sampling provides a natural model for the readout process. This has motivated studies of sequencing coverage~\cite{bodur2025random,bertuzzo2025coverage,abraham2024covering,tan2025labeled,barak2025conjectures} and of the expected number of samples required to recover a desired information file~\cite{gruica2024combinatorial,gruica2024geometry,boruchovsky2025making}.

Beyond retrieval efficiency, users may also wish to conceal which file they are interested in. Private information retrieval (PIR), introduced by Chor \emph{et al.}~\cite{chor1998private}, addresses the problem of retrieving a file without revealing its identity to the servers. Subsequent research has investigated communication efficiency, storage overhead, and the fundamental limits of PIR~\cite{dvir2016pir,fazeli2015codes,ghasemi2025improved,sun2017capacity,banawan2018capacity,zhang2019general}. However, most existing PIR schemes require direct access to specific files and allow servers to compute over the stored data to generate responses. Both assumptions are violated in many DNA-based storage systems, where data are stored as unordered DNA strands and algebraic operations are unavailable. Motivated by this, several PIR schemes tailored to DNA-based data storage have been proposed~\cite{wang2026sequential,erdougan2026private}. 

In~\cite{wang2026sequential}, the authors introduced the sequential random sampling DNA PIR (SRS DNA PIR) scheme for a single server. In their setting, a DNA pool is modeled as a random sampling server storing a database $\mathcal{X}$ of size $n$. Each query independently returns a file chosen uniformly at random from $\mathcal{X}$ with replacement. For a desired file $X$ chosen uniformly from $\mathcal{X}$, the scheme decides after each response whether to continue sampling or stop, and guarantees that $X$ is obtained before termination. The download cost is the expected number of queries, and the privacy leakage is the mutual information between the desired file identity and the response. For example, stopping as soon as $X$ is received has download cost $n$, but reveals $X$ as the last returned file. In contrast, collecting all files provides perfect privacy at a download cost of $nH_n$, where $H_n$ is the $n$-th harmonic number. These strategies illustrate the tradeoff between retrieval efficiency and user privacy.
Beyond a single server, Erdo{\u{g}}an \emph{et al.}~\cite{erdougan2026private} consider two servers, each storing the entire database. They propose PIR protocols with selective access via primers.

In this paper, we extend the random sampling model of~\cite{wang2026sequential} to \textit{multiple} DNA pools. The $n=m\alpha$ files are partitioned, without replication, among $m$ random sampling servers, each storing $\alpha$ files. 
 The user queries one server at a time and may adapt the server selection and stopping decisions to the observed responses. This can offer advantages over the single-server setting. For example, stopping as soon as $X$ is received requires download cost $\alpha$, while retrieving all files requires download cost $nH_{\alpha}$. However, the servers may collude to share their observations, which can increase privacy leakage. While collusion is well studied in classical PIR~\cite{sun2018capacity,freij2017private}, to the best of our knowledge, this is the first work that studies colluding servers in DNA-based PIR. We focus on the worst-case scenario where all servers collude, and measure privacy leakage with respect to the response across all servers. Unlike the single-server setting, privacy now depends not only on when sampling stops but also on which servers are queried.

We investigate the tradeoff between download cost and privacy leakage in this setting.
General lower bounds are established, including a tight bound for $\alpha = 2$. An optimal construction achieving the bound for $\alpha = 2$ is obtained. For arbitrary server sizes, a family of schemes is developed by applying single-server SRS DNA PIR schemes to randomly selected subsets of servers. Numerical comparisons illustrate the achievable tradeoffs and their gaps from the lower bounds.

The paper is organized as follows. Section~\ref{sec:background} presents the system model. Section~\ref{sec:lower_bound} derives lower bounds, and Section~\ref{sec:constructions} presents constructions. 
Section~\ref{sec:numerical_comparison} provides numerical comparisons. Section~\ref{sec:conclusion} concludes the paper. 

\section{System Model and Background}\label{sec:background}

For integers $m\leq n$, let $[m, n] = \{m, m+1, \ldots, n\}$ and $[n] = [1, n]$. For $\bx = (x_1, x_2, \ldots, x_n)$ and $A\subset [n]$, we denote by $\bx_{A}$ the subsequence of $\bx$ corresponding to the indices in $A$. 
Let $H(X)$, $H(X\mid Y)$, and $I(X; Y)$ denote the entropy, the conditional entropy, and the mutual information of random variables $X$ and $Y$.
For an integer $k > 0$, let $H_k = \sum_{i = 1}^k \frac{1}{i}$ be the $k$-th harmonic number and let $H_0 = 0$. 
For a real number $x \geq 0$, let $\bar{H}_x = \psi(x+1)+\gamma$ be the continuous harmonic number, where $\psi$ is the digamma function and $\gamma$ is the Euler-Mascheroni constant. It is increasing and concave in $x$, and $\bar{H}_k = H_k$ for every integer $k\geq 0$, see~\cite{abramowitz1965handbook}.

We consider a DNA-based data storage system in which the DNA strands are stored in multiple pools of DNA molecules, and the user can access the data by randomly sampling from each of the pools. 
Following~\cite{wang2026sequential}, we model each pool as a random sampling server.

\begin{definition}[Random sampling server~\cite{wang2026sequential}]
    A server storing $\alpha$ files is said to be a random sampling server if each query returns one of the $\alpha$ stored files independently and uniformly at random, with replacement.
\end{definition}

Thus, to retrieve a desired file, the user needs to query the server repeatedly until that file is returned. 
Example~\ref{exa:random_sampling_server} illustrates this process.

\begin{example}\label{exa:random_sampling_server}
    Suppose that a random sampling server stores three files $x_1$, $x_2$, and $x_3$, and that the user wishes to retrieve $x_3$. The returned files may be $x_2, x_1, x_1, x_2, x_3$. In this case, five queries are needed.
\end{example}

Besides retrieving the desired file, the user may want to keep the identity of the desired file private from the servers. 
Since some of the servers may collude to share information, we consider the worst case, in which all servers collude. 
Throughout the paper, a file is identified with its index, so privacy concerns the identity of the desired file and not its content. 
The following definition extends~\cite[Definition~2]{wang2026sequential} from one server to $m$ colluding servers.

\begin{definition}[SRS DNA PIR scheme]
    Consider a set of $n$ files $\cX = \{x_1, x_2, \ldots, x_n\}$ and $m$ random sampling servers, where server $j$ stores a subset $\cX_j\subseteq\cX$ and $\bigcup_{j\in[m]}\cX_j = \cX$. A sequential random sampling DNA private information retrieval (SRS DNA PIR) scheme $\cS = \cS(m, (\cX_j)_{j\in [m]})$ consists of the following parts:
        \begin{itemize}
        \item \textbf{Desired file:} The file $X \in \cX$ that the user wishes to retrieve.
        \item \textbf{Strategy:} The user queries one server at a time. At each step, the scheme $\cS$ determines, based on $X$ and on the files returned so far, whether to stop or which server to query next; this decision may be randomized. Let $\tau$ be the random variable denoting the stopping time, i.e., the number of queries made, and let $R = (r_1, r_2, \ldots, r_{\tau})$ denote the sequence of returned files.
        \item \textbf{Recovery:} Once the scheme terminates, $\cS$ outputs the desired file $X$ from $R$, i.e., $X \in R$.
    \end{itemize}
    The \textbf{evaluation metrics} of $\cS$ are defined as follows:
    \begin{itemize}
        \item \textbf{Download cost:} The expected number of queries until the scheme terminates: $E_{\cS} := \bbE[\tau]$.
        \item \textbf{Privacy leakage:} The amount of information that the colluding servers gain about $X$: $L_{\cS} := I(X; R)$.
    \end{itemize}
\end{definition}

We further assume that each file is equally likely to be the desired one, i.e., $X$ is uniformly distributed over $\cX$ and $H(X) = \log n$. We focus on the case where every file is stored on exactly one server and all servers have the same size $\alpha$. Thus, $n = m\alpha$ and $\cX_1,\ldots,\cX_m$ form a partition of $\cX$. Let $\cX_j = \{x_{\alpha (j - 1) + 1}, \ldots, x_{\alpha j}\}$ for $j\in[m]$. We identify server $j$ with its set of files $\cX_j$. 
We write $\cS(m, \alpha)$ for an SRS DNA PIR scheme in this setting.

\begin{example}\label{exa:SRS_DNA_PIR_scheme}
    Let $j_0\in[m]$ be the index of the server that stores the desired file, i.e., $X\in\cX_{j_0}$. Consider the following two SRS DNA PIR schemes. 
    \begin{itemize}
        \item \emph{Least-cost scheme $\cS_{\ell}$:} Query $\cX_{j_0}$ until $X$ is returned.
        \begin{itemize}
            \item Each query returns $X$ with probability $1/\alpha$. Therefore, $E_{\cS_{\ell}} = \alpha$.
            \item The scheme stops as soon as $X$ is returned, so $X$ is the last file in $R$. Hence, $R$ determines $X$, and $L_{\cS_{\ell}} = H(X) - H(X\mid R) = \log n$.
        \end{itemize}
        \item \emph{Full-download scheme $\cS_f$:} Query the servers one by one in a fixed order, each until all its $\alpha$ files have been returned.
        \begin{itemize}
            \item The expected number of queries to obtain all files of one server is $\alpha H_{\alpha}$. Therefore, $E_{\cS_f} = nH_{\alpha}$.
            \item The queries do not depend on $X$. Hence, $R$ is independent of $X$, and $L_{\cS_f} = 0$.
        \end{itemize}
    \end{itemize}
\end{example}

Example~\ref{exa:SRS_DNA_PIR_scheme} illustrates the tradeoff between the download cost and the privacy leakage. Our goal is to characterize this tradeoff. 

The \emph{download cost--privacy leakage region} $\cC(m, \alpha)$ is the set of pairs $(E_\cS, L_\cS)$ of all schemes $\cS = \cS(m, \alpha)$.

\begin{problem}\label{prob:main}
For given $m$ and $\alpha$, determine
\begin{align*}
L_{m, \alpha}(E) &= \inf\{L : (E, L) \in \cC(m, \alpha)\},\\
E_{m, \alpha}(L) &= \inf\{E : (E, L) \in \cC(m, \alpha)\}.
\end{align*}
\end{problem}

The single-server case was studied in~\cite{wang2026sequential}, where the following family of schemes was constructed.
\begin{theorem}[{\cite[Theorem~3]{wang2026sequential}}]\label{thm:randomized_DNA_PIR_single}
    For each $r\in[\alpha]$, let $\{p_i^r\}_{i=0}^{\alpha-r}$ be a probability distribution. Then there exists an SRS DNA PIR scheme $\cS_{\mathcal{R}} = \cS_{\mathcal{R}}(1, \alpha)$ with
    \begin{align*}
        E_{\cS_{\mathcal{R}}} &= \alpha + \sum_{r = 1}^{\alpha}\sum_{i = 0}^{\alpha - r}p_i^r \left(H_{\alpha-r} - H_{\alpha-r-i}\right),\\
        L_{\cS_{\mathcal{R}}} &= \log \alpha + \frac{1}{\alpha}\sum_{r = 1}^{\alpha}\sum_{j = 1}^{r}p_{r - j}^j\left(\log p_{r - j}^j - \log \sum_{\ell = 1}^{r} p_{r - \ell}^{\ell}\right).
    \end{align*}
\end{theorem}

\section{Lower Bounds on the Cost-Leakage Tradeoff}\label{sec:lower_bound}
In this section, we investigate lower bounds on the download cost and privacy leakage of an SRS DNA PIR scheme $\cS=\cS(m,\alpha)$. We start with a basic lower bound on the download cost.

\begin{claim}\label{cla:lower_bound_1}
    For any $m$, $\alpha$, and $0\leq L \leq \log n$, it holds that $E_{m, \alpha}(L) \geq \alpha$.
\end{claim}
\begin{IEEEproof}
    For any desired file $X = x_{\alpha (j - 1) + k}$ stored on server $j$, the user must query server $j$ until $X$ is returned. As the probability of obtaining $X$ from server $j$ in a single query is $1/\alpha$, the expected number of queries to obtain $X$ from server $j$ is $\alpha$. Therefore, the expected number of queries for any SRS DNA PIR scheme is at least $\alpha$, i.e., $E_{m, \alpha}(L) \geq \alpha$.
\end{IEEEproof}

To provide a more refined lower bound, we introduce some notation. For a scheme $\cS = \cS(m,\alpha)$, let $M$ be the random variable denoting the number of servers that are queried, and let $B_1, \ldots, B_M$ be the random variables denoting these servers, in the order in which they are first queried. Let $\bm{D} = (D_1, D_2, \ldots, D_m)$, where $D_j$ is the random variable denoting the number of distinct files obtained from the $j$-th queried server $B_j$, and $D_j = 0$ for $j>M$. Let
\[\cD_{m, \alpha}:=\bigcup_{s = 1}^m\{(d_1, \ldots, d_s, 0, \ldots, 0): d_1,\ldots, d_s\in [\alpha]\}.\]
The set $\cD_{m,\alpha}$ contains all possible values of $\bm{D}$. A vector $(d_1, \ldots, d_s, 0, \ldots, 0)\in\cD_{m,\alpha}$ describes the case where $s$ servers are queried, the $j$-th queried server returns $d_j$ distinct files, and the remaining $m-s$ servers are not queried.

Let $g(x) = \alpha(\bar{H}_\alpha - \bar{H}_{\alpha - x})$ for a real number $0\leq x\leq\alpha$. For an integer $k$, the value $g(k)$ is the expected number of queries needed to obtain $k$ distinct files from a server~\cite{daniella2023cover}. Since $\bar{H}_x$ is increasing and concave, $g$ is increasing and convex. Let $R^{B_j}$ be the sequence of files returned by server $B_j$, which is empty for $j>M$. Lemma~\ref{lem:number_of_queries_each_server} bounds the number of queries $|R^{B_j}|$ to server $B_j$ in terms of $D_j$.


\begin{lemma}\label{lem:number_of_queries_each_server}
    For any $m$, $\alpha$, SRS DNA PIR scheme $\cS = \cS(m, \alpha)$, and $j\in[m]$, it holds that $\bbE[|R^{B_j}|]\geq \bbE[g(D_j)]\geq g(\bbE[D_j])$.
\end{lemma}
\begin{IEEEproof}
    For any $t\geq 0$, let $T_j(t)$ be the random variable denoting the number of queries to the server $B_j$. Therefore, $t = \sum_{j = 1}^{m}T_j(t)$. Let $D_j(t)$ be the number of distinct files obtained from the server $B_j$ after $t$ queries. Consider the filtration $\cF_t$ where $\cF_t$ is the $\sigma$-algebra generated by the choice of servers, $X$, $t$, and the response from the servers. Let $Y_{j, t} = T_j(t) - g(D_{j}(t))$ and $A_{j, t}$ be the event that the $(t + 1)$-th query is to the server $B_j$. When $D_{j}(t) < \alpha$, we have
    \[\bbE[Y_{j, t + 1}|\cF_t, A_{j, t}^c] = Y_{j, t}\]
    and
    \begin{align*}
        &\bbE[Y_{j, t + 1}|\cF_t, A_{j, t}] \\
        = &\bbE[T_j(t) + 1 - g(D_{j}(t + 1))|\cF_t, A_{j, t}]\\
        = &T_j(t) + 1 - (g(D_{j}(t)) \\
        &+ \frac{\alpha - D_{j}(t)}{\alpha}(g(D_{j}(t) + 1) - g(D_{j}(t))))\\
        = &Y_{j, t} + 1 - \frac{\alpha - D_{j}(t)}{\alpha}(g(D_{j}(t) + 1) - g(D_{j}(t)))\\
        = &Y_{j, t} + 1 - \frac{\alpha - D_{j}(t)}{\alpha}\cdot \frac{\alpha}{\alpha - D_{j}(t)} = Y_{j, t}.
    \end{align*}
    If $D_{j}(t) = \alpha$, then each query to the server $B_j$ only increases $|R^{B_j}|$. Therefore, $Y_{j, t}$ is a submartingale. Note that $\bbE[|Y_{j, t + 1}-Y_{j, t}|\mid \cF_t]\leq \alpha$ and the stopping time $\tau$ satisfies $\bbE[\tau] = E_{\cS} < \infty$. By the optional stopping theorem~\cite[Theorem~4.8.5]{durrett2019probability}, we have $\bbE[|R^{B_j}| - g(D_{j})] \geq 0$. This implies that
    \[\bbE[|R^{B_j}|]\geq \bbE[g(D_{j})]\geq g(\bbE[D_j]),\]
    where the last inequality follows from the fact that $g(x)$ is convex and the Jensen's inequality. 
\end{IEEEproof}

Lemma~\ref{lem:number_of_queries_each_server} yields the following lower bound on the download cost.

\begin{theorem}\label{thm:lower_bound_2}
    For any $m$, $\alpha$, and $0\leq L \leq \log n$, it holds that $E_{m, \alpha}(L)\geq m\cdot g(\alpha 2^{-L})$.
\end{theorem}
\begin{IEEEproof}
    Consider an SRS DNA PIR scheme $\cS = \cS(m, \alpha)$ with $L_{\cS} = L$. By Lemma~\ref{lem:number_of_queries_each_server}, we have
    \begin{align*}
        E_{\cS} &= \bbE[|R|] = \sum_{j = 1}^{m}\bbE[|R^{B_j}|]\geq \sum_{j = 1}^{m}g(\bbE[D_j]).
    \end{align*}
    As $g(x)$ is convex, by the Jensen's inequality, we have
    \[E_{\cS}\geq mg\left(\alpha - \left(\frac{1}{m}\sum_{j = 1}^{m}\bbE[D_j]\right)\right).\]
    Furthermore, as $X$ is in the responses, we have 
    \[L_{\cS} = I(X; R) = H(X) - H(X|R)\geq \log n - \bbE\left[\log \sum_{j = 1}^{m} D_j\right].\]
    As $\log x$ is concave, we have
    \[\log \bbE\left[\sum_{j = 1}^{m} D_j\right]\geq \bbE\left[\log \sum_{j = 1}^{m} D_j\right]\geq \log n - L_{\cS},\]
    which implies that $\frac{1}{m}\sum_{j = 1}^{m}\bbE[D_j]\geq \alpha 2^{-L_{\cS}} = \alpha 2^{-L}$.
    In total, we have 
    \begin{align*}
        E_{\cS} \geq mg(\alpha 2^{-L}).
    \end{align*}
\end{IEEEproof}

Corollary~\ref{cor:lower_bound_3} follows from Claim~\ref{cla:lower_bound_1} and Theorem~\ref{thm:lower_bound_2}.
\begin{corollary}\label{cor:lower_bound_3}
        For any $m$, $\alpha$, and $0\leq L \leq \log n$, it holds that $E_{m, \alpha}(L)\geq \max\{\alpha,\, m\cdot g(\alpha 2^{-L})\}$.
\end{corollary}

Next, we provide a lower bound on the privacy leakage in terms of the distribution of $\bm{D}$. We start with some notation. 
For $j\in[M]$, let $Z_{j,1},\ldots,Z_{j,D_j}$ be the $D_j$ distinct files in $R^{B_j}$, listed in the order of their first appearance. There is a unique pair of random indices $(J,K)$ such that $X=Z_{J,K}$. In particular, $B_J$ is the server that stores $X$, and $K$ is the position of $X$ among the distinct files obtained from this server. Given $R$, the desired file $X$ and its index pair $(J,K)$ determine each other. Moreover, $\bm{D}$ is a function of $R$. Therefore, we have
\[H(X\mid R) = H(J,K\mid R) \leq H(J,K\mid\bm{D}).\]
Let $p_{\bd,j,k} = \Pr(\bm{D}=\bd,J=j,K=k)$, $\bm{p}_{\bd} = (p_{\bd,j,k})_{j,k}$, and $p_{\bd} = \Pr(\bm{D}=\bd) = \sum_{j, k}p_{\bd,j,k}$. The definition of conditional entropy gives
\[H(J,K\mid\bm{D}) = \sum_{\bd}p_{\bd}H(J,K\mid\bm{D} = \bd) = \sum_{\bd}\phi(\bm{p}_{\bd}),\]
where the sums are over $\bd\in\cD_{m,\alpha}$, and $\phi(\bv) = (\sum_{\ell=1}^{n}v_\ell)\log \sum_{\ell=1}^{n}v_\ell-\sum_{\ell=1}^{n}v_\ell\log v_\ell$ for any $\bv = (v_1, \ldots, v_n)$ with $v_\ell \geq 0$, where $0\log 0 = 0$. Therefore, $L_{\cS} \geq \log n - \sum_{\bd}\phi(\bm{p}_{\bd})$, and this bound depends on $\cS$ only through the probabilities $\{p_{\bd,j,k}\}$.

\begin{example}
    Consider the schemes $\cS_{\ell}$ and $\cS_f$ of Example~\ref{exa:SRS_DNA_PIR_scheme}.
    \begin{itemize}
        \item For $\cS_{\ell}$, only the server that stores $X$ is queried, so $J=1$ and $\bm{D} = (D_1, 0, \ldots, 0)$. The scheme stops when $X$ is returned, so $X$ is the last of the $D_1$ distinct files, i.e., $K=D_1$. This results in $p_{(k, 0, \ldots, 0), 1, k} = 1/\alpha$ for all $k\in [\alpha]$ and $p_{\bd,j,k} = 0$ for other $(\bd,j,k)$, which gives $H(J,K\mid \bm{D}) = 0$.
        \item For $\cS_f$, the user queries all servers until all files are returned, so $\bm{D} = (\alpha, \ldots, \alpha)$. This results in $p_{(\alpha, \ldots, \alpha), j, k} = 1/n$ for all $j\in [m]$ and $k\in[\alpha]$ and $p_{\bd,j,k} = 0$ for other $(\bd,j,k)$, which gives $H(J,K\mid \bm{D}) = \log n$.
    \end{itemize}
\end{example}

Lemmas~\ref{lem:index_desired_file} and~\ref{lem:probability_distribution_constraints} provide constraints on $\{p_{\bd,j,k}\}$.

\begin{lemma}\label{lem:index_desired_file}
For any SRS DNA PIR scheme $\cS = \cS(m,\alpha)$, $j\in[m]$, and $k\in[\alpha]$, it holds that $\Pr(J=j,K=k) = \Pr(J=j)/\alpha$.
\end{lemma}
\begin{IEEEproof}
Fix a desired file $X=x$, and let $B_j$ be the server containing $x$. As each sampling in a server is uniformly at random with replacement, conditioning on $J=j$ does not change the distribution of the samples of $B_j$. Therefore, the probability that $x$ is the $k$-th distinct file in the sequence is
\[\Pr(K = k|X = x,\hspace{-0.5pt} J = j)\hspace{-0.5pt} =\hspace{-0.5pt}\frac{1}{\alpha-k+1}\prod_{\ell=1}^{k-1}\frac{\alpha-\ell}{\alpha-\ell+1}\hspace{-0.5pt} =\hspace{-0.5pt} \frac{1}{\alpha},\]
where an empty product equals $1$. This also implies that $\Pr(K=k\mid J=j)$ is independent of $x$, and thus
\[\Pr(J=j, K=k) = \frac{\Pr(J=j)}{\alpha}.\]
\end{IEEEproof}

Let $a_j = \Pr(J=j)$. By Lemma~\ref{lem:index_desired_file}, we have $\sum_{\bd\in\cD_{m,\alpha}}p_{\bd, j, k} = a_j/\alpha$ for all $j\in[m]$ and $k\in[\alpha]$.

\begin{lemma}\label{lem:probability_distribution_constraints}
For any SRS DNA PIR scheme $\cS=\cS(m,\alpha)$ with $E_{\cS}\leq E$, the probabilities $\{a_j\}$, $\{p_{\bd,j,k}\}$ induced by $\cS$ satisfy
\begin{equation}\label{eq:download_cost_constraint}
\begin{aligned}
    &a_j\geq 0,~\sum_{j=1}^{m}a_j=1,~p_{\bd,j,k}\geq 0,~p_{\bd,j,k} = 0~\forall k>d_j,\\
    &\sum_{\bd\in\cD_{m,\alpha}}p_{\bd,j,k}=\frac{a_j}{\alpha},~\sum_{\bd\in\cD_{m,\alpha}}\sum_{j=1}^{m}\sum_{k=1}^{\alpha}p_{\bd,j,k}=1,\\
    &\sum_{\bd\in\cD_{m,\alpha}}p_{\bd}\sum_{i=1}^{m}g(d_i) \leq E,
\end{aligned}
\end{equation}
for all $\bd\in\cD_{m,\alpha}$, $j\in[m]$, and $k\in[\alpha]$.
\end{lemma}
\begin{IEEEproof}
The nonnegativity of $a_j$ and $p_{\bd,j,k}$ follows from the definitions. Since $J\in[m]$, we also have $\sum_{j=1}^{m}a_j = 1$. As any DNA PIR scheme must ensure that the user can obtain the desired file, we have that $K\leq D_J$, which gives $p_{\bd, j, k}=0$ whenever $k>d_j$.

By Lemma~\ref{lem:index_desired_file}, we have
\[\sum_{\bd\in\cD_{m,\alpha}}p_{\bd,j,k} = \Pr(J=j,K=k) = \frac{a_j}{\alpha},\]
which also implies that
\[\sum_{\bd\in\cD_{m,\alpha}} \sum_{j=1}^{m}\sum_{k=1}^{\alpha}p_{\bd,j,k} = \sum_{j=1}^{m}\sum_{k=1}^{\alpha}\frac{a_j}{\alpha} = 1.\]

To prove the cost constraint, applying Lemma~\ref{lem:number_of_queries_each_server} to each server gives
\begin{align*}
E\geq E_{\cS} &= \sum_{j=1}^{m}\bbE\bigl[|R^{B_j}|\bigr] \geq \sum_{j=1}^{m}\bbE[g(D_j)] \\
&= \bbE\left[\sum_{j=1}^{m}g(D_j)\right] = \sum_{\bd\in\cD_{m,\alpha}}p_{\bd}\sum_{i=1}^{m}g(d_i).
\end{align*}
\end{IEEEproof}

We now formulate a lower bound on the privacy leakage in terms of the distribution $\{p_{\bd,j,k}\}$.

\begin{theorem}\label{thm:lower_bound_minimization}
For any $m$, $\alpha$, and $E\geq\alpha$, it holds that $L_{m,\alpha}(E)\geq\underline{L}_{m,\alpha}(E)$, where
\begin{equation}\label{eq:coverage_minimization}
\underline{L}_{m,\alpha}(E) = \min\Big\{\log n- \sum_{\bd\in\cD_{m,\alpha}}\phi(\bm{p}_{\bd})\Big\},
\end{equation}
and the minimum is over all $\{a_j\}$, $\{p_{\bd,j,k}\}$ satisfying~\eqref{eq:download_cost_constraint}.
\end{theorem}

\begin{IEEEproof}
For any SRS DNA PIR scheme $\cS = \cS(m,\alpha)$ with $E_\cS = E$, its induced probabilities $\{a_j\}$ and $\{p_{\bd,j,k}\}$ are feasible for \eqref{eq:download_cost_constraint} by Lemma~\ref{lem:probability_distribution_constraints}. Therefore, we have
\begin{align*}
    L_{\cS} &=  H(X)-H(X\mid R) \ge \log n-H(J,K\mid \bm{D})\\
            & = \log n -\sum_{\bd\in\cD_{m,\alpha}}\phi(\bm{p}_{\bd}) \ge \underline{L}_{m,\alpha}(E).
\end{align*}
\end{IEEEproof}

In Section~\ref{subsec:optimal_construction_alpha_2}, we show that the lower bound in Theorem~\ref{thm:lower_bound_minimization} is tight for $\alpha = 2$. However, the number of variables in \eqref{eq:coverage_minimization} is of order $m\alpha^{m+1}$, which makes the optimization problem hard to solve for large values of $m$ and $\alpha$.

\section{Constructions of SRS DNA PIR Schemes}\label{sec:constructions}
In this section, we provide two constructions of SRS DNA PIR schemes. The first one shows that the lower bound in Theorem~\ref{thm:lower_bound_minimization} is tight for $\alpha = 2$, and the second one works for any integers $m$ and $\alpha$.

\subsection{Optimal Scheme for Servers of Size Two}\label{subsec:optimal_construction_alpha_2}

Let $\alpha = 2$. For any $\{a_j\}$ and $\{p_{\bd,j,k}\}$ satisfying \eqref{eq:download_cost_constraint}, we construct an SRS DNA PIR scheme whose privacy leakage equals the objective value in \eqref{eq:coverage_minimization}.

\begin{construction}\label{cons:optimal_scheme_alpha_2}
Let $\alpha=2$, and let $\{a_j\}$ and $\{p_{\bd,j,k}\}$ satisfy~\eqref{eq:download_cost_constraint}. For $X\in\cX_{j_0}$, the scheme works as follows.
\begin{enumerate}
    \item Sample $J=j$ with probability $a_j$. Then choose a permutation $\sigma=(\sigma(1),\ldots,\sigma(m))$ of $[m]$ uniformly at random among those with $\sigma(j)=j_0$. The $i$-th queried server $B_i$ is server $\sigma(i)$. In particular, $B_j$ is server $j_0$, which stores $X$.
    \item Query each of the servers $B_1,\ldots,B_j$ once, in this order. If $B_j$ returns $X$, set $k=1$; otherwise, set $k=2$.
    \item Draw $\bd\in\cD_{m,2}$ with probability $2p_{\bd,j,k}/a_j$.
    \item Let $\bd=(d_1,\ldots,d_s,0,\ldots,0)$ be the drawn vector. Query each of the servers $B_{j+1},\ldots,B_s$ once, in this order. Then, for $i=1,\ldots, s$ in increasing order, if $d_i=2$, query $B_i$ until its second file is returned.
\end{enumerate}
\end{construction}
\begin{example}
    Let $m = 2$ and $\alpha = 2$, so that server $1$ stores $x_1,x_2$ and server $2$ stores $x_3,x_4$. Let $a_1 = a_2 = 1/2$, $p_{(2, 1), 1, 1} = p_{(1, 2), 1, 1} = p_{(2, 1), 2, 1} = p_{(1, 2), 2, 1} = 1/8$, and $p_{(2, 1), 1, 2} = p_{(1, 2), 2, 2} = 1/4$. These probabilities satisfy \eqref{eq:download_cost_constraint} with $E=4$.
    Assume $X = x_1\in\cX_{1}$. The scheme in Construction~\ref{cons:optimal_scheme_alpha_2} works as follows:
    \begin{enumerate}
        \item The user samples $J = 2$, with probability $a_2=1/2$. Then $\sigma = (2, 1)$, since $\sigma(2)=j_0=1$. Thus, $B_1$ is server $2$ and $B_2$ is server $1$.
        \item The user queries $B_1$ and receives $x_3$, and then queries $B_2$ and receives $x_1$. Since $X = x_1$, set $k=1$.
        \item Since $2p_{(2, 1),2,1}/a_2 = 2p_{(1, 2),2,1}/a_2 = 1/2$, the user draws $\bd = (2, 1)$ or $\bd=(1,2)$ with probability $1/2$ each. Assume $\bd = (2,1)$.
        \item Since $d_1 = 2$, the user queries $B_1$ until $x_4$ is returned, and then stops.
    \end{enumerate}
    If $B_2$ returns $x_2$ in Step~2, then $k=2$, and the user draws $\bd=(1,2)$ with probability $2p_{(1,2),2,2}/a_2=1$. In Step~4, the user then queries $B_2$ until $x_1$ is returned.
\end{example}

\begin{theorem}
\label{lem:performance_optimal_scheme_alpha_2}
The scheme $\cS$ in Construction~\ref{cons:optimal_scheme_alpha_2} is well-defined, and its download cost and privacy leakage are
\[E_{\cS} = \sum_{\bd\in\cD_{m,2}}p_{\bd}\sum_{i=1}^{m}g(d_i)~\text{and}~L_{\cS}=\log n-\sum_{\bd\in\cD_{m,2}}\phi(\bm{p}_{\bd}).\]
\end{theorem}

\begin{IEEEproof}
By the constraints~\eqref{eq:download_cost_constraint}, for any $a_j>0$, we have $\sum_{\bd\in\cD_{m,2}}\frac{2p_{\bd,j,k}}{a_j}=1$. Hence, the distribution used in Step~3 is normalized. Furthermore, any vector selected with positive probability satisfies $d_j\geq k$, which ensures that the desired file is obtained. Therefore, the scheme is well-defined.

Next we claim that the prescribed distribution is realized. For $a_j>0$, it holds that
\begin{align*}
    &\Pr(\bm{D} = \bd, J = j, K = k) \\
    = &\Pr(J = j)\hspace{-0.37pt}\Pr(K = k\mid J = j)\hspace{-0.37pt}\Pr(\bm{D} = \bd\mid J = j,K = k)\\
    = &a_j\cdot\frac{1}{2}\cdot\frac{2p_{\bd,j,k}}{a_j} \\
    = &p_{\bd, j, k}.
\end{align*}
If $a_j=0$, both sides are zero by the constraints, which proves the claim.

For the download cost, note that when $\alpha=2$, we have $g(0)=0$, $g(1)=1$, and $g(2)=3$. Fix a terminal coverage vector $\bd=(d_1,\ldots,d_s,0,\ldots,0)$ with positive probability. The first pass uses exactly $s$ queries, one for each visited server. For every position $i$ with $d_i=2$, the expected number of queries needed to obtain the second distinct file is $2$. Consequently, we have 
\[\bbE[\tau\mid\bm{D}=\bd] = s + 2|\{i:d_i=2\}| = \sum_{i=1}^{m}g(d_i).\] Averaging over $\bm{D}$ gives
\[ E_{\cS} = \sum_{\bd\in\cD_{m,2}}p_{\bd}\sum_{i=1}^{m}g(d_i). \]

Finally, we compute the privacy leakage. Fix a response sequence $\br=(r_1,\ldots,r_\tau)$ with positive probability, and let $\bd=(d_1, \ldots, d_s, 0, \ldots, 0)$ be the corresponding terminal coverage vector. For $j\in[s]$ and $k\in[d_j]$, 
let $z_{j, k}$ be the $k$-th distinct file returned by the $j$-th queried server in $\br$.
Consider the probability of $R = \br$ when $X=z_{j,k}$ and $a_j>0$. The event $R = \br$ happens if and only if the following occur in order:
\begin{itemize}
    \item The user selects $J = j$ and $\sigma(i)$ is the $i$-th queried server in $\br$. The probability of this event is $a_j \frac{(m-s)!}{(m-1)!}$, since the first $s$ positions of the permutation are fixed by the response sequence.
    \item The user selects $\bm{D} = \bd$. The probability of this event is $2p_{\bd, j, k}/a_j$.
    \item The user receives the prescribed responses from the servers. The probability of this event is $2^{-\tau}$, as the size of each server is $2$.
\end{itemize}
In total, we have
\begin{align*}
    \Pr(R=\br\mid X=z_{j,k})&= a_j \frac{(m-s)!}{(m-1)!} \frac{2p_{\bd,j,k}}{a_j} 2^{-\tau}\\ 
    &= 2^{1-\tau}\frac{(m-s)!}{(m-1)!}p_{\bd,j,k}.
\end{align*} 
Here, the factor $2^{1-\tau}\frac{(m-s)!}{(m-1)!}$ 
does not depend on $(j,k)$. Since $X$ is uniform over the $n$ files, it follows that
\begin{align*}
    \Pr(R=\br) &= \frac{1}{n}2^{1-\tau}\frac{(m-s)!}{(m-1)!}\sum_{u=1}^{s}\sum_{v=1}^{d_u}p_{\bd,u,v}\\
    &=\frac{p_{\bd}}{n}2^{1-\tau}\frac{(m-s)!}{(m-1)!}.
\end{align*}
By Bayes' rule, we have
\[\Pr(X=z_{j,k}\mid R=\br) = \frac{2^{1-\tau}\frac{(m-s)!}{(m-1)!}\frac{p_{\bd,j,k}}{n}}{2^{1-\tau}\frac{(m-s)!}{(m-1)!}\frac{p_{\bd}}{n}}=\frac{p_{\bd,j,k}}{p_{\bd}},\]
which implies that
\[H(X\mid R=\br) = H(J,K\mid R=\br) = \frac{\phi(\bm{p}_{\bd})}{p_{\bd}}.\]
Averaging over the response sequences yields that
\[H(X\mid R) = \sum_{\substack{\bd\in\cD_{m,2}\\p_{\bd}>0}}p_{\bd}\frac{\phi(\bm{p}_{\bd})}{p_{\bd}}=\sum_{\bd\in\cD_{m,2}}\phi(\bm{p}_{\bd}),\]
which gives the privacy leakage
\[ L_{\cS} = H(X)-H(X\mid R) = \log n- \sum_{\bd\in\cD_{m,2}}\phi(\bm{p}_{\bd}).\]
\end{IEEEproof}

\begin{remark}
The key property used in Construction~\ref{cons:optimal_scheme_alpha_2} is that a single query to the server holding the desired file suffices to determine $K$. If the desired file is returned, then $K=1$; otherwise, $K=2$. 
For $\alpha>2$, this property does not hold: after the first query, the user only learns whether $K = 1$ or $K\in [2,\alpha]$.
\end{remark}

\subsection{Construction for Servers of Any Size}

Construction~\ref{cons:optimal_scheme_alpha_2} works only for $\alpha = 2$, and its optimal probabilities $\{p_{\bd,j,k}\}$ are obtained from the optimization problem in~\eqref{eq:coverage_minimization}, which we solve numerically in Section~\ref{sec:numerical_comparison}. In this subsection, we construct SRS DNA PIR schemes for any integers $m$ and $\alpha$ from a single-server SRS DNA PIR scheme.

\begin{construction}[$s$-subset scheme]\label{cons:multiple_from_single}
    Let $s\in [m]$, and let $\cS = \cS(1,\alpha)$ be a single-server SRS DNA PIR scheme. For the desired file $X\in\cX_{j_0}$, the scheme works as follows.
    \begin{enumerate}
        \item Choose a subset $T\subseteq [m]$ of size $s$ uniformly at random among all sets of size $s$ that contain $j_0$.
        \item For each $j\in T$, in increasing order of $j$, apply the scheme $\cS$ to server $j$. The desired file is $X$ if $j=j_0$, and a file of server $j$ chosen uniformly at random otherwise.
    \end{enumerate}
\end{construction}

The order in Step~2 must not depend on $j_0$. Otherwise, the order reveals information about the server that stores $X$. Lemma~\ref{lem:construction_multiple_servers_based_on_single_server} gives the download cost and the privacy leakage of the $s$-subset scheme of Construction~\ref{cons:multiple_from_single}.

\begin{lemma}\label{lem:construction_multiple_servers_based_on_single_server}
    If $\cS=\cS(1, \alpha)$ satisfies $E_\cS = E$ and $L_\cS = L$, then for every $s\in[m]$ the $s$-subset scheme $\cS'=\cS'(m, \alpha)$ of Construction~\ref{cons:multiple_from_single} satisfies $E_{\cS'} = sE$ and $L_{\cS'} = L+\log \frac{m}{s}$.
\end{lemma}

\begin{IEEEproof}
    The scheme $\cS$ is applied to $s$ servers, so $E_{\cS'} = sE$. For the privacy leakage, note that $R$ reveals the set $T$. Since $T$ is uniform among the sets of size $s$ that contain $j_0$, and $\cS$ is applied to each server in $T$ with a uniformly distributed desired file, the index $j_0$ is uniformly distributed over $T$ given $R$. Thus, $H(j_0\mid R)=\log s$. Given $j_0$, the file $X$ depends on $R$ only through the response of server $j_0$, so $H(X\mid j_0,R) = \log\alpha - L$. Therefore, $H(X\mid R) = \log s+\log\alpha-L$, and $L_{\cS'} = \log (m\alpha) - H(X\mid R) = L+\log\frac{m}{s}$.
\end{IEEEproof}

\begin{example}
    Let $m = \alpha$, and let $\cS_{\ell} = \cS_{\ell}(1, \alpha)$ and $\cS_{f} = \cS_{f}(1, \alpha)$ be the schemes of Example~\ref{exa:SRS_DNA_PIR_scheme} for a single server. Thus, $E_{\cS_\ell} = \alpha$, $L_{\cS_\ell}= \log \alpha$, $E_{\cS_f} = \alpha H_\alpha$, and $L_{\cS_f} = 0$. By Lemma~\ref{lem:construction_multiple_servers_based_on_single_server}, Construction~\ref{cons:multiple_from_single} gives the following two schemes for $m$ servers.
    \begin{itemize}
        \item With $s = 1$ and $\cS=\cS_f$, the user downloads all files of the server that stores $X$. The download cost is $\alpha H_\alpha$, and the privacy leakage is $0+\log m = \log \alpha$.
        \item With $s = m$ and $\cS=\cS_{\ell}$, the user applies $\cS_\ell$ to every server. The download cost is $m\alpha = \alpha^2$, and the privacy leakage is $\log \alpha+\log\frac{m}{m} = \log \alpha$.
    \end{itemize}
    The two schemes have the same privacy leakage, but the first one has a smaller download cost, since $H_\alpha<\alpha$ for $\alpha\geq 2$.
\end{example}


Corollary~\ref{cor:construction_multiple_servers_based_on_single_server} follows from Lemma~\ref{lem:construction_multiple_servers_based_on_single_server} and Theorem~\ref{thm:randomized_DNA_PIR_single}.
\begin{corollary}\label{cor:construction_multiple_servers_based_on_single_server}
    For any $\alpha$, $m$, and $s\in[m]$, and for each $r\in[\alpha]$, let $\{p_i^r\}_{i=0}^{\alpha-r}$ be a probability distribution. Then there exists an SRS DNA PIR scheme $\cS'_{\mathcal{R}} = \cS'_{\mathcal{R}}(m, \alpha)$ with
    \begin{align*}
        E_{\cS'_{\mathcal{R}}} &= s\Big(\alpha + \sum_{r = 1}^{\alpha}\sum_{i = 0}^{\alpha - r}p_i^r \left(H_{\alpha-r} - H_{\alpha-r-i}\right)\Big),\\
        L_{\cS'_{\mathcal{R}}} &= \log \frac{m\alpha}{s} + \frac{1}{\alpha}\sum_{r = 1}^{\alpha}\sum_{j = 1}^{r}p_{r - j}^j(\log p_{r - j}^j - \log \sum_{\ell = 1}^{r} p_{r - \ell}^{\ell}).
    \end{align*}
\end{corollary}

As in~\cite[Theorem~4]{wang2026sequential}, two schemes can be mixed. The download cost of the resulting scheme is the convex combination of the two download costs, and its privacy leakage is at most the convex combination of the two privacy leakages.
\begin{lemma}\label{lem:mixing_schemes}
    Let $\cS_1 = \cS_1(m, \alpha)$ and $\cS_2 = \cS_2(m, \alpha)$ be two SRS DNA PIR schemes with $(E_{\cS_1}, L_{\cS_1})=(E_1, L_1)$ and $(E_{\cS_2}, L_{\cS_2})=(E_2, L_2)$. For any $\lambda\in [0, 1]$, let $\cS_{\lambda}$ be the scheme in which the user runs $\cS_1$ with probability $\lambda$ and $\cS_2$ otherwise. Then it follows that
    \[E_{\cS_{\lambda}}=\lambda E_1+(1-\lambda)E_2~\text{and}~L_{\cS_{\lambda}}\le\lambda L_1+(1-\lambda)L_2.\]
\end{lemma}

\begin{IEEEproof}
    Let $C\in\{1,2\}$ be the choice of the SRS DNA PIR schemes of the user. The cost follows from linearity of expectation. Since conditioning reduces entropy and $C$ is independent of $X$, we have
    \begin{align*}
        L_{\cS_{\lambda}} &= H(X)-H(X\mid R)\\
        &\leq H(X)-H(X\mid R,C)\\
        &=\lambda I(X;R\mid C=1)+(1-\lambda)I(X;R\mid C=2)\\
        &=\lambda L_1+(1-\lambda)L_2,
    \end{align*}
    where the last equality holds because, given $C=c$, the user runs exactly $\cS_c$. 
\end{IEEEproof}

Corollary~\ref{cor:bridging_s} follows directly from Lemma~\ref{lem:construction_multiple_servers_based_on_single_server} and Lemma~\ref{lem:mixing_schemes}.

\begin{corollary}\label{cor:bridging_s}
    Let $\cS = \cS(1,\alpha)$ be an SRS DNA PIR scheme with $E_\cS = E$ and $L_\cS = L$. For any $s \in [m-1]$ and $E^* \in [sE, (s+1)E]$, let $\lambda = (s+1) - \frac{E^*}{E}$. Then there exists an SRS DNA PIR scheme $\cS^* = \cS^*(m, \alpha)$ with $E_{\cS^*} = E^*$ and $L_{\cS^*} \le \log \frac{m2^L}{s^{\lambda}(s+1)^{1-\lambda}}$.
\end{corollary}
Corollary~\ref{cor:bridging_s} is needed when the choices of $s$ leave a gap in the achievable download cost. For example, with $\alpha = 2$, Theorem~\ref{thm:randomized_DNA_PIR_single} gives download costs $E \in [2,3]$ only. Lemma~\ref{lem:construction_multiple_servers_based_on_single_server} with $s=1$ then reaches $E_{\cS'}\in[2,3]$, and with $s=2$ reaches $E_{\cS'}\in[4,6]$, leaving the interval $(3,4)$ uncovered. Corollary~\ref{cor:bridging_s} covers this gap.
\section{Numerical Comparison}\label{sec:numerical_comparison}
In this section, we present numerical comparisons between the achievable schemes and the lower bounds. For a fixed $s$, the $s$-subset schemes in Corollary~\ref{cor:construction_multiple_servers_based_on_single_server} have download cost between $s\alpha$ and $s\alpha H_\alpha$,  where the distributions $\{p_i^r\}$ are chosen numerically to minimize the privacy leakage. Schemes with different values of $s$ can be mixed by Lemma~\ref{lem:mixing_schemes}.

For $\alpha = 2$, the lower bound in Theorem~\ref{thm:lower_bound_minimization} is tight, and the schemes in Construction~\ref{cons:optimal_scheme_alpha_2} attain it. Hence, the optimal tradeoff is known, but computing it requires solving the optimization problem in~\eqref{eq:coverage_minimization} numerically. Fig.~\ref{fig:comparison_alpha_2} compares it, for $m = 3$, with the two results that avoid this optimization: the lower bound in Corollary~\ref{cor:lower_bound_3} and the $s$-subset schemes. The lower bound in Corollary~\ref{cor:lower_bound_3} is not tight, whereas the mixed $s$-subset schemes are close to the optimal tradeoff.
\begin{figure}[t]
    \centering
    \includegraphics[width=\columnwidth]{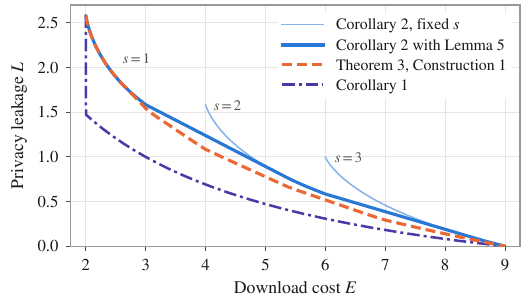}
    \caption{Privacy leakage versus download cost for $\alpha = 2$ and $m = 3$.}
    \label{fig:comparison_alpha_2}
\end{figure}

For $\alpha > 2$, the tightness of the lower bound in Theorem~\ref{thm:lower_bound_minimization} remains unknown, and Construction~\ref{cons:optimal_scheme_alpha_2} is not applicable. Fig.~\ref{fig:comparison_alpha_5} compares, for $\alpha = 5$ and $m = 4$, the $s$-subset schemes with the lower bounds in Theorem~\ref{thm:lower_bound_minimization} and Corollary~\ref{cor:lower_bound_3}. The lower bound in Theorem~\ref{thm:lower_bound_minimization} is tighter than the one in Corollary~\ref{cor:lower_bound_3}, and the mixed $s$-subset schemes are close to it. Since the optimal tradeoff lies between these schemes and the lower bound in Theorem~\ref{thm:lower_bound_minimization}, it is determined up to a small gap for these parameters.
\begin{figure}[t]
    \centering
    \includegraphics[width=\columnwidth]{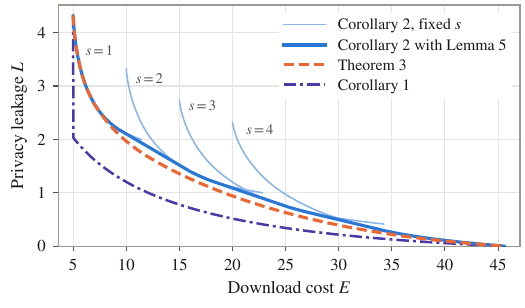}
    \caption{Privacy leakage versus download cost for $\alpha = 5$ and $m = 4$.}
    \label{fig:comparison_alpha_5}
\end{figure}



\section{Conclusion and Extensions}\label{sec:conclusion}
We studied the tradeoff between the download cost and the privacy leakage of SRS DNA PIR schemes with multiple colluding random sampling servers of equal size. We derived lower bounds on this tradeoff, one of which is tight when each server stores two files, and constructed schemes for servers of any size that are close to the lower bound in our numerical examples.

The optimal tradeoff for servers that store more than two files remains open. Other directions are partial collusion, where only some subsets of servers share their observations, and servers of unequal sizes or with overlapping contents.


\section*{AI Statement}
The authors used Claude (Anthropic) and ChatGPT (OpenAI) in the preparation of this paper. In Section~\ref{sec:numerical_comparison}, Claude helped to write the code for the numerical evaluation and to produce Figs.~\ref{fig:comparison_alpha_2} and~\ref{fig:comparison_alpha_5}. In the other sections, both tools were used to revise the wording of text written by the authors. The authors take full responsibility for the content of the paper, including all results, constructions, and proofs.

\bibliographystyle{IEEEtran}
\bibliography{ref.bib} 

\end{document}